\documentclass{article}
\usepackage{etex}
\reserveinserts{28}
\usepackage{setspace, amsmath, amssymb, amsthm}
\usepackage[colorlinks = true]{hyperref}
\usepackage[margin = 2.5cm]{geometry}
\usepackage[english]{babel}
\usepackage[utf8x]{inputenc}
\usepackage{amsmath}
\usepackage{graphicx}
\usepackage[usenames,dvipsnames]{pstricks}
\usepackage{epsfig}
\usepackage{pst-grad} 
\usepackage{pst-plot} 
\usepackage{epstopdf}
\usepackage{subcaption}
\usepackage{tikz}
\usetikzlibrary{arrows}
\usetikzlibrary{calc}
\usetikzlibrary{fit}
\usepackage{pstricks-add}
\newtheorem{theorem}{Theorem}[section]
\newtheorem{lemma}[theorem]{Lemma}
\newtheorem{proposition}[theorem]{Proposition}
\newtheorem{corollary}[theorem]{Corollary}

\usepackage{tikz}

\usepackage[normalem]{ulem}
\usepackage{color}

\title{Eigenvalues of weakly balanced signed graphs and graphs with negative cliques}
\author{Ranveer Singh \thanks{ Department of Mathematics, Technion-Israel Institute of Technology Haifa, Israel - 31000. email: \texttt{ranveer@iitj.ac.in}}\\ R. B. Bapat \thanks{Stat-Math Unit, Indian Statistical Institute Delhi, 7-SJSS Marg, New Delhi - 110 016. email: \texttt{rbb@isid.ac.in}}}

\begin{document}
        \maketitle
%
    
\begin{abstract}   
In a signed graph $G$, an induced subgraph is called a negative clique if it is a complete graph and all of its edges are negative. In this paper, we give the characteristic polynomials and the eigenvalues of some signed graphs having negative cliques. This includes cycle graphs, path graphs,  complete graphs with vertex-disjoint negative cliques of different orders, and star block graphs with negative cliques. Interestingly, if we reverse the signs of the edges of these graphs, we get the families of weakly balanced signed graphs, thus the eigenvalues of wide classes of weakly balanced signed graphs are also calculated. In social network theory, the eigenvalues of the signed graphs play an important role in determining their stability and developing the measures for the degree of balance. 
\end{abstract}

\emph{keywords:} Signed graph, weakly balanced graph, linear subdigraphs, negative clique. 

\section{Introduction}
In 1956, Cartwright and Frank Harary modeled the cognitive structure of balance in signed social networks by introducing the concept of signed graphs \cite{b19,b21}.  In a signed social graph, the vertices represent individuals and a positive edge (the edge with a positive sign) between two vertices reflects the existence of liking relationship, whereas, a negative edge (the edge with a negative sign) represents disliking. After the introduction of signed graphs, several attempts have been made for investigating a possible connection between the eigenvalues and balance of signed graphs, for example, see \cite{singh2017measuring, kunegis2014applications,b27,singh2017mathcal}. 

A graph $G$ consists of a finite set of vertices $V(G)$ and set of edges $E(G)$ consisting of distinct, unordered pairs of vertices. Thus, $(i,j)$ or $(j,i)$ represents an edge between vertices $i,j \in V(G)$, and $i$, $j$ are called adjacent vertices. The number of vertices in $G$ is called its order. If $G$ is equipped with a weight function $f: E(G) \rightarrow \{-1,0,1\}$, then $G$ is called a signed graph. Thus, a signed graph may have positive, negative edges with weights $1$, $-1$, respectively. 
Let $G$ be a signed graph on $n$ vertices.   Then, the adjacency matrix $A$ of order $n\times n$ associated with $G$ is defined by $$A_{i,j}=\begin{cases}
1 & \mbox{if the vertices $i,j$ are linked with a positive edge,}\\
-1& \mbox{if the vertices $i,j$ are linked with a negative edge,}\\
0 & \mbox{if the vertices $i,j$ are not linked,}
\end{cases} $$ where, $1\leq i,j\leq n$. The eigenvalues of $G$ are the eigenvalues of its adjacency matrix $A$.
The degree of a vertex $i$ in $G$ is defined as $d_i=\sum_j |A_{i,j}|$. Thus, it equals to the number of incident edges to $i$, irrespective of its signs.


\begin{figure}
    \begin{subfigure}[b]{0.24\textwidth}
        \centering
        
            \begin{tikzpicture}[scale=0.35]

                     \draw  node[draw,circle,scale=0.5] (1) at (4,0) {$v_1$};
                    
                     \draw  node[draw,circle,scale=1](2) at (0,4) {2};
                     
                     \draw  node[draw,circle,scale=1](3) at (-4,0) {3};
                     
                     \draw  node[draw,circle,scale=1](4) at (0,-4) {4};
 
  \tikzset{edge/.style = {{dashed}= latex'}}                  
     \draw[edge] (1) to (2);
      \draw[edge] (3) to (4);
                        
           \tikzset{edge/.style = {- = latex'}}
           \draw[edge] (1) to (4);  
  \draw[edge] (2) to (3); 
                     
      \end{tikzpicture}
        
        \caption{Balanced $C_4$}
        \label{bal}
    \end{subfigure}
    \begin{subfigure}[b]{0.24\textwidth}
    \centering
        
            \begin{tikzpicture} [scale=0.35]

                    \draw  node[draw,circle,scale=1] (1) at (4,0) {1};
                   
                    \draw  node[draw,circle,scale=1](2) at (0,4) {2};
                    
                    \draw  node[draw,circle,scale=1](3) at (-4,0) {3};
                    
                    \draw  node[draw,circle,scale=1](4) at (0,-4) {4};
    \tikzset{edge/.style = {{dashed}= latex'}}                  
    \draw[edge] (2) to (3);                     
          \tikzset{edge/.style = {- = latex'}}
          \draw[edge] (1) to (4); 
          \draw[edge] (1) to (2);
     \draw[edge] (3) to (4);         
            \end{tikzpicture}
        
        \caption{Unbalanced $C_4$}   
        \label{unbal}
    \end{subfigure}
    \begin{subfigure}[b]{0.24\textwidth}
        \centering
       
            \begin{tikzpicture} [scale=0.4]
                
                \tikzset{VertexStyle/.style = {shape = circle,fill = black,minimum size = 19pt}}
                \draw  node[draw,circle,scale=1] (1) at (4,0) {1};
                \draw  node[draw,circle,scale=1] (2) at (2.8284,2.8284) {2};
                \draw  node[draw,circle,scale=1](3) at (0,4) {3};
                \draw  node[draw,circle,scale=1](4) at (-2.8284,2.8284) {4};
                \draw  node[draw,circle,scale=1](5) at (-4,0) {5};
                \draw  node[draw,circle,scale=1](6) at (-2.8284,-2.8284) {6};
                \draw  node[draw,circle,scale=1](7) at (0,-4) {7};
                \draw  node[draw,circle,scale=1](8) at (2.8284,-2.8284) {8};

\tikzset{edge/.style = {{dashed}= latex'}} 
  \draw[edge] (2) to (3);
   \draw[edge] (2) to (4); 
  \draw[edge] (4) to (3);
\draw[edge] (6) to (7);
   \draw[edge] (6) to (8); 
  \draw[edge] (7) to (8);  

  \tikzset{edge/.style = {- = latex'}}
 \draw[edge] (1) to (2);
   \draw[edge] (1) to (3); 
  \draw[edge] (1) to (4);
\draw[edge] (1) to (5);
   \draw[edge] (1) to (6); 
  \draw[edge] (1) to (7);
\draw[edge] (1) to (8);

\draw[edge] (2) to (5);
  \draw[edge] (2) to (6); 
  \draw[edge] (2) to (7);
\draw[edge] (2) to (8);

\draw[edge] (3) to (5);
  \draw[edge] (3) to (6); 
  \draw[edge] (3) to (7);
\draw[edge] (3) to (8);

\draw[edge] (4) to (5);
  \draw[edge] (4) to (6); 
  \draw[edge] (4) to (7);
\draw[edge] (4) to (8);

 \draw[edge] (5) to (6); 
  \draw[edge] (5) to (7);
\draw[edge] (5) to (8);

            \end{tikzpicture}
        
        \caption{$K^{2,3}_8$}
        \label{comp}
    \end{subfigure}       
    \begin{subfigure}[b]{0.24\textwidth}
            \centering
            
    \begin{tikzpicture} [scale=0.35]
                    
                    \draw  node[draw,circle,scale=1] (1) at (4,0) {1};
                    \draw  node[draw,circle,scale=1] (2) at (2.8284,2.8284) {2};
                    \draw  node[draw,circle,scale=1](3) at (0,4) {3};
                    \draw  node[draw,circle,scale=1](4) at (-2.8284,2.8284) {4};
                    \draw  node[draw,circle,scale=1](5) at (-4,0) {5};
                    \draw  node[draw,circle,scale=1](6) at (-2.8284,-2.8284) {6};
                    \draw  node[draw,circle,scale=1](7) at (0,-4) {7};
                    \draw  node[draw,circle,scale=1](8) at (2.8284,-2.8284) {8};
                    \draw  node[draw,circle,scale=1](9) at (0,0) {9};

\tikzset{edge/.style = {{dashed}= latex'}} 
  \draw[edge] (4) to (3);
   \draw[edge] (3) to (9); 
  \draw[edge] (4) to (9);
\draw[edge] (9) to (7);
   \draw[edge] (9) to (8); 
  \draw[edge] (7) to (8);  

  \tikzset{edge/.style = {- = latex'}}
 \draw[edge] (1) to (2);
   \draw[edge] (2) to (9); 
  \draw[edge] (1) to (9);
\draw[edge] (9) to (5);
   \draw[edge] (5) to (6);
   \draw[edge] (9) to (6);                        
                \end{tikzpicture} 
    
    \caption{3-regular star block graph}
            \label{star}
        \end{subfigure}

\caption{Examples: The dotted lines show negative edges (weight -1).} 
\label{example}
\end{figure}
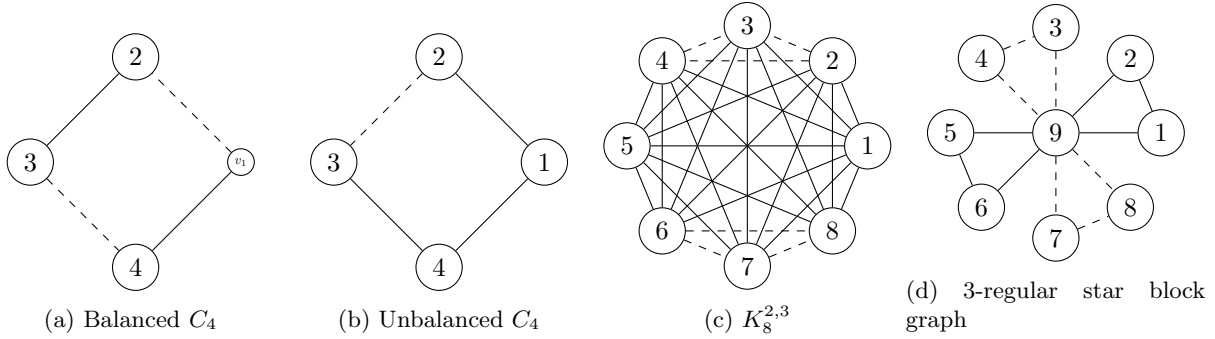

 We denote a cycle graph on $n$ vertices by $C_{n}$ or $n$-cycle. The adjacency matrix $A$ of $C_{n}$ is given by $A_{i,i+1}=A_{i+1,i}\in \{1,-1\},$ $i=1,2, \hdots, n-1$ and $A_{n,1}=A_{1,n}\in \{1,-1\}$, all other entries of $A$ are zero. Moreover, the sign of $C_n$ is defined as the product of signs of its edges. If the sign of $C_n$ is positive it is called balanced cycle, otherwise, it is called an unbalanced cycle  \cite{b21}. Examples of a balanced and an unbalanced $C_4$ are shown in Figure (\ref{bal}), (\ref{unbal}), respectively. We denote a  tree on $n$ vertices by $T_n.$ The path graph on $n$ vertices is denoted by $P_n$. \\

Let $G$ be a signed graph. If each cycle of $G$ is balanced, then $G$ is called a balanced signed graph, otherwise, an unbalanced signed graph. A tree is balanced. The balanced signed graphs display an interesting graph partitioning phenomenon as stated by the following theorem. 

\begin{theorem} \cite{b17}
A signed graph $G$ is balanced if and only if, either all of its edges are positive or the vertices can be partitioned into two subsets such that each positive edge joins vertices
in the same subset and each negative edge joins vertices in different subsets.
\end{theorem}

In 1967, Davis \cite{b20} gave a generalization of balanced signed graphs, which are known as weakly balanced signed graphs. A signed graph is called a weakly balanced graph if and only if, either all of its edges are positive or the vertices can be partitioned into $k\geq 2$ vertex subsets such that each positive edge joins vertices
in the same subset and each negative edge joins vertices in different subsets. The necessary and sufficient condition for a signed graph to be weakly balanced is that it should not have any cycle with exactly one negative edge \cite{b20}.

 When each edge of a clique is negative we call it a negative clique. Similarly, if each edge of a clique is positive, then we call it a positive clique. We denote a complete graph on $n$ vertices having each edge positive,  by $K_n$. By $K^{m,r}_n$, we denote a  complete graph on $n$ vertices having a $m$   vertex-disjoint negative cliques each of order $r$, and all the other edges positive except those are in the negative cliques. Example of a $K^{2,3}_8$ graph is given in Figure \ref{comp}, where two vertex-disjoint negative cliques, each of order 3 are on vertex-sets \{2,3,4\} and \{6,7,8\}, respectively. We also consider the complete graphs having vertex-disjoint negative cliques of different orders such that the negative cliques cover the whole vertex-set.  A block in a signed graph $G$ is a maximal subgraph which has no cut-vertex. If each block of $G$ is a complete graph, then $G$ is called block graph. For block graphs without negative edges see \cite{bapat2014adjacency}. If block graph $G$ has at most one cut-vertex, then we call it star block graph. We consider a star block graph having $k$  blocks each having $r$  vertices. We call it a $r$-regular star block graph. An example of a $3$-regular star block graph is given in Figure \ref{star}.  It is to be noted that, for all the above signed graphs (except $C_n$ with exactly one positive edge), if we reverse the signs of their edges, we get weakly balanced signed graphs. Thus negative of the eigenvalues of the above graphs give the eigenvalues of corresponding weakly balanced graphs.    

 The characteristic polynomial of a square matrix $A$ of order $n$ is the polynomial defined by
$\det \left( A-\lambda I \right),$
where $I$ denotes the $n\times n$ identity matrix. We denote the characteristic polynomial of $A$ by $\phi(A)$. The characteristic polynomial of signed graph $G$, denoted by $\phi(G)$, is characteristic polynomial of its adjacency matrix $A$ that is $\phi(G)=\phi(A)$. The eigenvalues of a matrix $A$ are roots of the characteristic polynomial $\det(A-\lambda I)$. The spectrum of a signed graph $G$ is set of the eigenvalues of its adjacency matrix along with their multiplicities.  For convenience, we can relabel the vertices in graph $G$. In graph theory, these relabelling are captured by permutation similarity of adjacency matrix $A$. The determinant of permutation matrices is equal to $\pm 1$. Thus, relabelling on vertex-set keep the determinant, and characteristic polynomial unchanged. Not to mention that the eigenvalues of signed $C_n$ are given in literature \cite{germina2010signed, germina2011products} by different proof techniques. Here we give their characteristic polynomial using the matching concept, the eigenvalues and the determinant can be easily deduced from it.

\subsection{Matchings and Coates digraph}
First, we modify some preliminaries from \cite{b4} for signed graphs. A matching in a signed graph $G$ is a collection of edges no two of which have a vertex in common. The largest number of edges in a matching in $G$ is the matching number $m(G)$. A matching with $k$ edges is called a $k$-matching. A perfect matching of $G$ also called a 1-factor, is a matching that covers all vertices of $G$. 

The Coates digraph $D(A)$ generated from a matrix $A$ of order $n$ has $n$ vertices labelled by $1, 2,\hdots, n$ and for each pair of such vertices $i, j$ a directed edge exists from $j$ to $i$ of weight $A_{i,j}$ \cite{b4}. The elements of the main diagonal of $A$ corresponds to loops at vertices in $D(A)$. If diagonal elements of $A$ are zero, then no loops are considered on corresponding vertices of $D(A).$   A linear subdigraph of $D(A)$ is a spanning subdigraph of $D(A)$ in which each vertex has indegree 1 and outdegree 1 that is exactly one edge into each vertex and exactly one (possibly the same, in the case of the loop) out of each vertex. Thus a linear subdigraph consists of a spanning collection of pairwise vertex-disjoint cycles. The weight of a linear subdigraph is the product of the weights of the edges in it. For example, the Coates digraph representation of the matrix $A=\begin{bmatrix}
 a_{11} & a_{12} \\
 a_{21}& a_{22}
\end{bmatrix}$ is given in Figure \ref{fig2a}.

By the Coates digraph of a signed graph, we mean the Coates digraph corresponding to the adjacency matrix of the signed graph. Consider a signed graph $G$ and denote its Coates digraph by $D(G)$. For an edge between vertices $i, j$ in $G$, there are two directed edges of equal weights in $D(G)$, one from $i$ to $j$ and other from $j$ to $i$. This forms a directed cycle of length $2$ which we call a directed $2$-cycle. In a linear subdigraph of $D(G)$,  $k$ such directed $2$-cycles appear due to the $k$ matchings in $G$. Thus, there is a one-one correspondence between matchings in $G$ and directed $2$-cycles in a linear subdigraph of $D(G)$. Thus, by $k$-matching in linear subdigraphs we mean, the existence of $k$ vertex-disjoint directed $2$-cycles.  For example, the Coates digraph of balanced $C_4$ in Figure \ref{bal} is shown in Figure \ref{fig2b}. Note that, there are two 2-matchings in balanced $C_n$ in Figure \ref{bal}. These are \{$(1,2),(3,4)$\}, and \{$(2,3),(1,4)$\}. In Figure \ref{fig2b}, corresponding to these two matchings, there are two directed 2-cycles in linear subdigraph $L_3$, $L_4$, respectively, in the Coates digraph of the balanced $C_n$.   Now we recall the definition of the determinant of the adjacency matrix $A$ of $G$ in terms of its linear subdigraphs in $D(G)$.

\begin{figure}
\centering
\begin{subfigure}{.5\textwidth}
  \centering
  \includegraphics[width=0.9\linewidth]{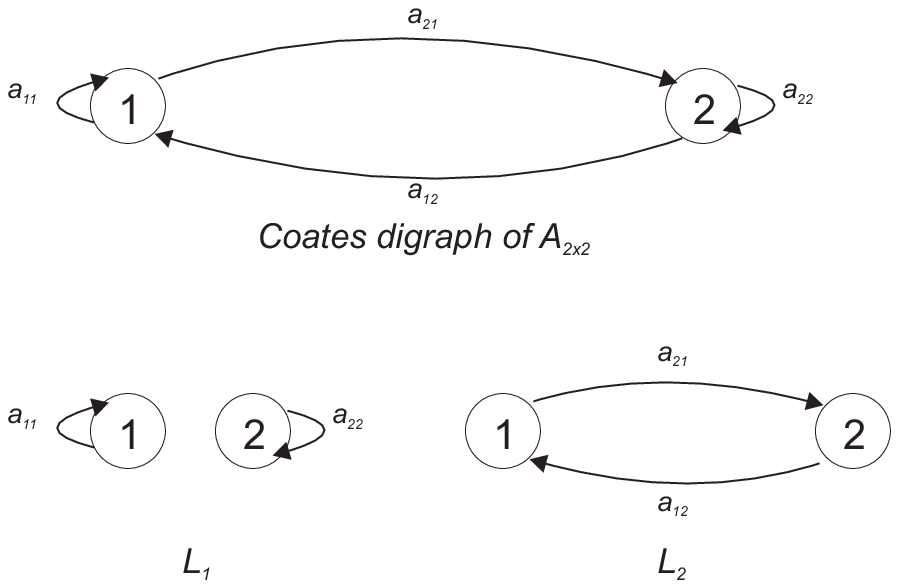}
  \caption{}
  \label{fig2a}
\end{subfigure}%
\begin{subfigure}{.5\textwidth}
  \centering
  \includegraphics[width=0.9\linewidth]{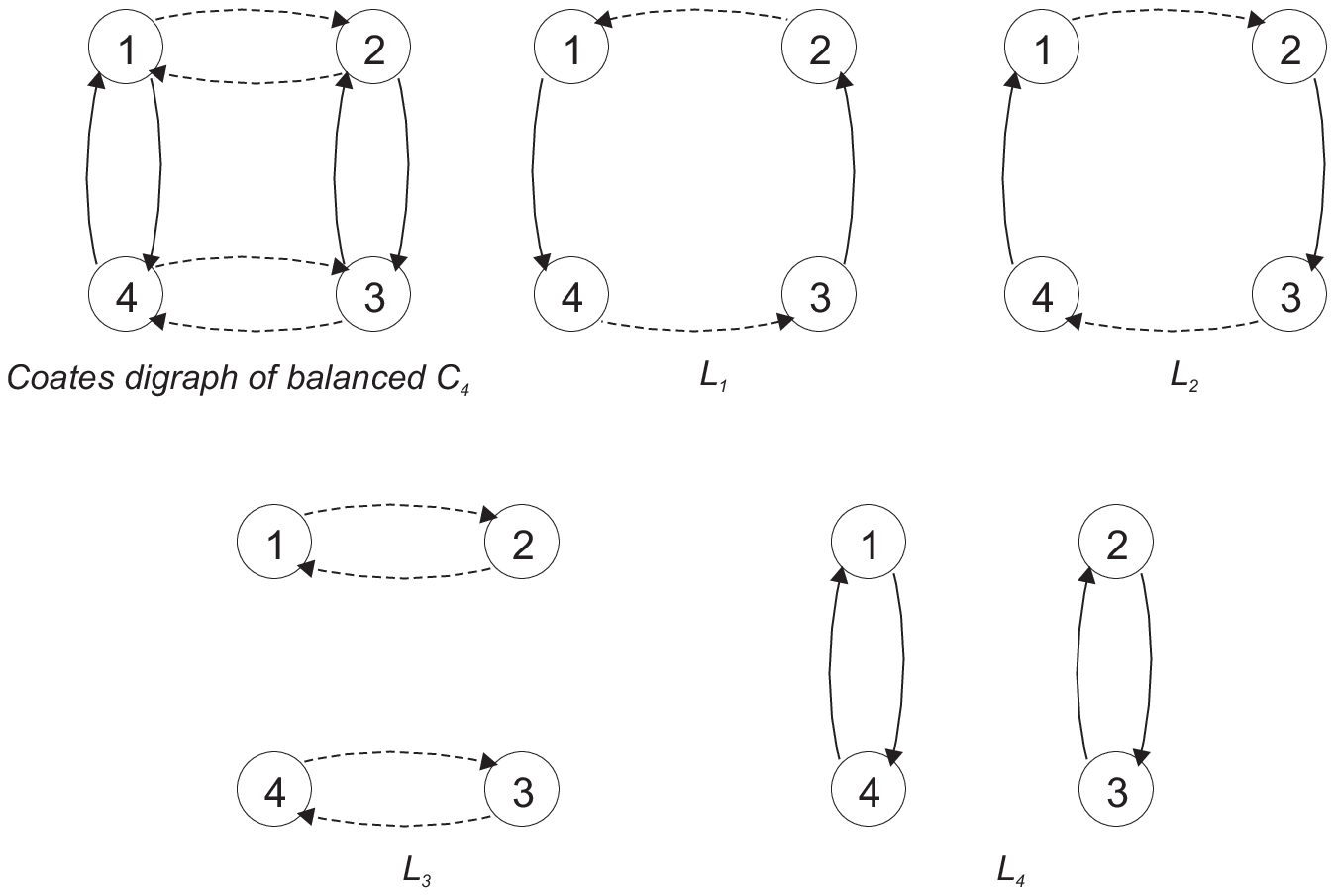}
  \caption{} 
  \label{fig2b}
\end{subfigure}
\caption{The Coates digraph the and linear subdigraphs of (a) $A_{2\times 2}$ where $L_1$ and $L_2$ are the linear subdigraphs (b) Balanced $C_4,$  where $L_1, L_2, L_3$, and $L_4$ are the linear subdigraphs.}
\label{fig2}
\end{figure}

\begin{theorem}\cite{b4} \label{thm1}
Let $A$ be a square matrix of order $n$. Then 
$$\det A = (-1)^n \sum_{L\in\mathcal{L}(A)}(-1)^{c(L)}w(L),$$ 
where, $w(L)$ is the weight of linear subdigraph $L$ of the Coates digraph $D(A)$, $c(L)$ is the number of directed cycles in $L$, and $\mathcal{L}(A)$ denotes the set of all linear subdigraphs of $D(A)$.
\end{theorem}

The paper is organized as the following. In Section \ref{path&cycle}, we calculate the characteristic polynomials, hence the eigenvalues and the determinant of cycle  and path graphs using the concept of linear subdigraphs and matching.  We calculate the characteristic polynomial, determinant and the eigenvalues of $K_n^{m,r}$ in Section \ref{cgnk}. In Section \ref{cgnkdo}, we give the bounds of the eigenvalues of complete graphs having disjoint negative cliques of different orders which cover the whole vertex-set. Finally, in Section  \ref{sbg} we calculate the eigenvalues of regular star block graphs. We again mention that the negative of the eigenvalues of the graphs in these sections gives the eigenvalues of wide classes of weakly balanced signed graphs.

\section{Characteristic Polynomial of $C_n$ and $P_n$} \label{path&cycle}
We denote the weight of the cycle  graph $C_n$ by $\delta$. If $C_n$ is balanced, $\delta=1$, otherwise, $\delta=-1$. The Coates digraph corresponding to the adjacency matrix $\Big(A(C_n)-\lambda I_n\Big)$ is a directed graph or digraph on $n$ vertices with 
\begin{enumerate}
\item a loop of weight $-\lambda$ at each vertex.
\item for each pair of adjacent vertices in cycle $C_n$, there are two opposite directed edges connecting these adjacent vertices in the Coates digraph.
\end{enumerate}   

Next, we require the number of $k$-matchings in $C_n$, which is used to find the linear subdigraphs of  the Coates digraph of $\Big(A(C_n)-\lambda I_n\Big)$. We state the following standard result \cite{match}. 
\begin{proposition} \label{matchincycle}
The number of $k$-matching in cycle graph $C_{n}$ is equal to
\begin{equation}\label{eqn2}
\frac{n}{n-k}\binom{n-k}{k}.
\end{equation}
\end{proposition}

For cycle graphs $m(G) = \left \lfloor n/2 \right \rfloor$, thus the number of all possible matching in $G$ is given by
\begin{equation}\label{eqn3}
\sum_{k=0}^{\left \lfloor n/2 \right \rfloor}\frac{n}{n-k}\binom{n-k}{k},
\end{equation}

where $k=0$ corresponds to no matching. Each $k$-matching in $C_n$ corresponds to $k$ vertex-disjoint directed $2$-cycles in its Coates digraph covering $2k$ vertices. These $k$ directed $2$-cycles along with the loops at the remaining $n-2k$ vertices form the linear subgraphs in the Coates digraph of $C_n$.

\begin{theorem}\label{cpc}
The characteristic polynomial $\phi(C_n)$ of cycle  graph  $C_n$, having weight $\delta\in\{-1,1\}$ is given by
\begin{equation*}
\phi(C_n)=\begin{cases}
(-1)^{n}\Bigg(\Big(\sum_{k=1}^{(\frac{n}{2}-1)}\frac{n}{n-k}\binom{n-k}{k}\times (-1)^{n-k}\times (-\lambda)^{n-2k}\Big)+2(-1)^{\frac{n}{2}}-2\delta\Bigg) & \mbox{if $n$ is even},\\
(-1)^{n}\Bigg(\sum_{k=1}^{\left \lfloor n/2 \right \rfloor}\frac{n}{n-k}\binom{n-k}{k}\times (-1)^{n-k}\times (-\lambda)^{n-2k}-2\delta\Bigg) & \mbox{if $n$ is odd.}
\end{cases}
\end{equation*}
\end{theorem}
\begin{proof}
In the Coates digraph of the matrix $\Big(A(C_n)-\lambda I_n\Big)$ there will be the following two type of the linear subdigraphs along with their contribution to $\phi(C_n).$
\begin{enumerate}
\item The two directed $n$-cycles; one clockwise and another  anticlockwise, respectively, each having weight $\delta$. Using Theorem \ref{thm1} their contribution to $\phi(C_n)$ is $$(-1)^n\Big(2(-1)^1\delta\Big)=(-1)^n(-2\delta).$$ 

\item The linear subdigraph having $k$-matching covering $2k$ vertices, and the loops at the remaining $n-2k$ vertices for $k=1,2,\hdots,\lfloor n/2\rfloor $. The weight of each $k$-matching is $1$, and the weight of the $n-2k$ loops is $(-\lambda)^{n-2k}$. The total number of cycles are $k+n-2k=n-k.$ If,
\begin{enumerate}
\item $n$ is even: for $k=\frac{n}{2}$, there will be two linear subdigraphs having $\frac{n}{2}$ directed $2$-cycles. Thus, no loop will be selected in these two linear subdigraphs. Their contribution is $$(-1)^n2(-1)^{\frac{n}{2}}.$$
\item $n$ is odd:  there will be no linear subdigraphs having $\frac{n}{2}$ directed $2$-cycles.
\end{enumerate} Thus, using Proposition \ref{matchincycle}, and combining 1. and 2., the result follows.
\end{enumerate} 
\end{proof}

\begin{corollary}
The determinant of cycle  $C_n$, having weight $\delta\in\{-1,1\}$ is given by
$$\det (C_n)=\begin{cases}
2-2\delta &\mbox{if $n$ is even and even multiple of 2,}\\
 -2-2\delta &\mbox{if $n$ is even and odd multiple of 2,}\\
 2\delta &\mbox{if $n$ is odd.}
\end{cases}$$
\end{corollary}
\begin{proof}
To calculate the determinant we need to set $\lambda=0$ in the characteristic polynomial. Hence, the result directly follows by Theorem \ref{cpc}.
\end{proof}

\subsection{Eigenvalues of $C_n$}
Let us consider a matrix $Q$ of order $n\geq2$ such that, the entry $Q_{i,i+1}\in \{1,-1\}$, $i=1,2,\hdots,n-1$, the entry $Q_{n,1}\in \{1,-1\}$ and the remaining entries of $Q$ are zero. The Coates digraph $D(Q-\lambda I)$ is a digraph having directed $n$-cycle with a loop of weight $-\lambda$ at each of its vertices. Thus, the Coates digraph $D(Q-\lambda I)$ has only two linear subdigraphs. One having the directed $n$-cycle without loops, and another consisting of the all $n$ loops. The weight of the directed $n$-cycle is either $1$ or $-1$. It follows from Theorem (\ref{thm1}) that the characteristic equation of $Q$ is given by:
\begin{equation}
(-1)^{n}\Big((-1)^{n}(-\lambda)^{n}+(-1)^{1}\delta\Big)=0 \implies \lambda^{n}-\delta=0,
\end{equation} 

which means that the eigenvalues of $Q$ are $1,\omega,\omega^{2},...\omega^{n-1}$, where,
$$\omega=\begin{cases}
e^{\frac{2\pi \iota}{n}} &\mbox{if $\delta =1$,}\\
 e^{\iota\frac{\pi +2\pi k}{n}} &\mbox{if $\delta =-1$.}
\end{cases}$$

For  a cycle $C_n$, the  adjacency matrix $A(C_n)=Q+Q'=Q+Q^{n-1}$ is a polynomial in $Q$ \cite{bapat2014adjacency}. Thus, the eigenvalues of $A(C_n)$ are obtained by evaluating the same polynomial at each of the eigenvalues of $Q$, thus the eigenvalues of $A(C_n)$ are $\omega^{k}+\omega^{n-k},k=1,\hdots,n$.

\begin{theorem}
The eigenvalues of  $C_n$ are
$$\omega=\begin{cases}
2\cos\frac{2\pi k}{n} &\mbox{if $C_n$ is balanced,}\\
 2\cos(\frac{\pi +2\pi k}{n}) &\mbox{if $C_n$ is unbalanced,}
\end{cases}$$ $k=1,2,\hdots,n$.
\end{theorem}
\begin{proof}
It is clear that the eigenvalues of $A(C_n)$ are $\omega^{k}+\omega^{n-k},k=1,...n.$ To derive the adjacency matrix of balanced $C_n$ from $Q$, the value of $\delta$ has to be 1. Similarly,  to derive the adjacency matrix of unbalanced $C_n$ from $Q$, the value of $\delta$ has to be $-1$. Now, if $\delta=1,$ $$\omega^{k}+\omega^{n-k}=\omega^{k}+\omega^{-k}
=e^{\frac{2\pi \iota k}{n}}+e^{-\frac{2\pi \iota k}{n}}
=2\cos\frac{2\pi k}{n},$$ 
if $\delta=-1,$
$$
e^{\iota\frac{\pi +2\pi k}{n}}+e^{-\iota\frac{\pi +2\pi k}{n}}=2\cos(\frac{\pi +2\pi k}{n}),
$$
for $k=1,2,\hdots,n.$
\end{proof}

%
%
%

\begin{theorem}
Let $\lambda_{1}\geq\lambda_{2}\geq \hdots\geq\lambda_{n}$ be the eigenvalues of a balanced cycle  graph and $\beta_{1}\geq\beta_{2}\geq \hdots \geq\beta_{n}$ be the eigenvalues of unbalanced cycle  graph of length $n>2.$ Then,
\begin{equation*}
\lvert\lambda_{i}-\beta_{i}\rvert=\lvert\lambda_{n-i+1}-\beta_{n-i+1}\rvert, i =1,2,\hdots, n.
\end{equation*}
\end{theorem}
\begin{proof}
\begin{enumerate}
\item If $n$ is even: $\cos$ function lie in range [-1 1]. The eigenvalues of a balanced and an unbalanced $C_n$ are $2\cos(\frac{2\pi k}{n})$, and $2\cos(\frac{\pi +2\pi k}{n})$, respectively, for $k=1,2,\hdots,n$. To get $\lambda_{1}\geq\lambda_{2}\geq \hdots \geq\lambda_{n}$ and $\beta_{1}\geq\beta_{2}\geq \hdots \geq\beta_{n}$ we need to sort the values of $2\cos(\frac{2\pi k}{n})$ and $2\cos(\frac{\pi +2\pi k}{n})$ in a descending order. 
Also,
$2\cos(\frac{2\pi k}{n})=2\cos(\frac{2\pi (n-k)}{n})$, and $2\cos(\frac{\pi +2\pi k}{n})=2\cos(\frac{-\pi -2\pi k}{n})=2\cos(\frac{\pi +2\pi (n-k-1)}{n})$. The sorted order of the eigenvalues of balanced $C_n$ is for the sequence $k=n,1,(n-1),2,(n-2),\hdots,i,(n-i),\hdots, (n/2+1),n/2$. For unbalanced $C_n,$ the sorted order is for the sequence $k=n,(n-1),1,(n-1-1),2,(n-2-1),\hdots,i,(n-i-1),\hdots, (n/2-1),n/2$. Now, consider $\lambda_i$ and $\lambda_{n-i+1}$. As their corresponding $k$ indices are at a difference of $n/2$, we have, $2\cos(\frac{2\pi (k\pm n/2)}{n})=-2\cos(\frac{2\pi k}{n})$. Hence, $\lambda_{n-i+1}=-\lambda_i$. Corresponding $k$ indices of $\beta_i$ and $\beta_{n-i+1}$ are also at a difference of $n/2$. Thus, $2\cos(\frac{\pi +2\pi (k\pm n/2)}{n})=-2\cos(\frac{\pi +2\pi k}{n})$. Hence, $\beta_{n-i+1}=-\beta_i$, and $\lvert\lambda_{i}-\beta_{i}\rvert=\lvert\lambda_{n-i+1}-\beta_{n-i+1}\rvert, i =1,2,\hdots, n.$

\item If $n$ is odd: following the similar steps as in the case for even $n$, in this case to get $\lambda_{1}\geq\lambda_{2}\geq \hdots \geq\lambda_{n}$, we need the sequence $k=n,1,(n-1),2,(n-2),\hdots,i,(n-i),\hdots, (n-1)/2,(n+1)/2$, and to get $\beta_{1}\geq\beta_{2}\geq \hdots \geq\beta_{n}$ we need the sequence $k=n,(n-1),1,(n-1-1),2,(n-2-1),\hdots,i,(n-i-1),\hdots, (n+1)/2,(n-1)/2.$  The difference between the $k$ index for $\beta_{i}$ and the $k$ index for $\alpha_{n-i+1}$ is $\pm n/2 -1/2$. We have, $2\cos(\frac{\pi +2\pi (k\pm n/2 -1/2)}{n})=-2\cos(\frac{2\pi k}{n})$. Hence, $\lambda_i=-\beta_{n-i+1}$. Similarly, $\beta_i=-\lambda_{n-i+1},$ thus $\lvert\lambda_{i}-\beta_{i}\rvert=\lvert\lambda_{n-i+1}-\beta_{n-i+1}\rvert, i =1,2,\hdots, n.$
 
\end{enumerate}
\end{proof}

\subsection{Characteristic Polynomial of $P_n$}\label{signedpath}
The Coates digraph corresponding to the adjacency matrix  $A(P_n)$ of  path graph $P_n$, is a directed graph having $n$ vertices with 
\begin{enumerate}
\item a loop of weight $-\lambda$ at each vertex.
\item for every pair of adjacent vertices in path $P_n$, there are two opposite directed edges, connecting these adjacent vertices in the Coates digraph.
\end{enumerate}   

We state the following standard result \cite{match}.
\begin{proposition} \label{matchinp}
The number of $k$-matching in path graph $P_{n}$ is equal to
\begin{equation}
\binom{n-k}{k}.
\end{equation}
\end{proposition}

Thus, for path graphs $m(G) = \left \lfloor n/2 \right \rfloor$, the number of all possible matchings in $G$ is given by:
\begin{equation}
\sum_{k=0}^{\left \lfloor n/2 \right \rfloor}\binom{n-k}{k}.
\end{equation}

\begin{theorem}\label{cpp}
The characteristic polynomial $\phi(P_n)$ of  $P_n$ is given by
\begin{equation*}
\phi(P_n)=\begin{cases}
(-1)^{n}\Bigg(\Big(\sum_{k=1}^{(\frac{n}{2}-1)}\binom{n-k}{k}\times (-1)^{n-k}\times (-\lambda)^{n-2k}\Big)+(-1)^{\frac{n}{2}}\Bigg) &\mbox{if $n$ is even,}\\
 (-1)^{n}\Bigg(\sum_{k=1}^{\left \lfloor n/2 \right \rfloor}\binom{n-k}{k}\times (-1)^{n-k}\times (-\lambda)^{n-2k}\Bigg) &\mbox{if $n$ is odd.}
\end{cases}
\end{equation*}
\end{theorem}

\begin{proof}
In the Coates digraph of the matrix $\Big(A(P_n)-\lambda I\Big)$ there will be the following type of linear subdigraph along with its contribution to $\phi(P_n)$. The subdigraph having $k$-matching covering $2k$ vertices and the loops at the remaining $n-2k$ vertices for $k=1,2,\hdots,\lfloor n/2\rfloor $. The weight of $k$-matching is $1$, and the weight of $n-2k$ loops is $(-\lambda)^{n-2k}$. The total number of cycles are $k+n-2k=n-k.$  If,
\begin{enumerate}
\item $n$ is even: for $k=\frac{n}{2}$, there will be one linear subdigraphs having $\frac{n}{2}$ directed $2$-cycles. Thus, no loop will be selected in this linear subdigraph. Its contribution is $$(-1)^n(-1)^{\frac{n}{2}}.$$
\item $n$ is odd:  There will be no linear subdigraphs having $\frac{n}{2}$ directed $2$-cycles.
\end{enumerate} Thus, using Proposition \ref*{matchinp}, and combining 1. and 2.,  the result follows.
\end{proof}
As the characteristic polynomial of all path graphs $P_n$ for a given $n$ is same, their eigenvalues are same. These can be found in \cite{bapat2010graphs}.

\begin{corollary}
The determinant of path $P_n$ is given by 
$$\det (P_n)=\begin{cases}
1 &\mbox{if $n$ is even and even multiple of 2,}\\
 -1 &\mbox{if $n$ is even and odd multiple of 2,}\\
 0 &\mbox{if $n$ is odd.}
\end{cases}$$
\end{corollary}
\begin{proof}
Proof directly follows using Theorem \ref{cpp} on setting $\lambda=0.$ \end{proof}

\section{Characteristic polynomial of $K^{m,r}_n$} \label{cgnk}
In this section we derive the characteristic polynomial of $K^{m,r}_{n}$. Here, the determinant and the eigenvalues are readily follows from the characteristic polynomial, hence they are stated as corollaries without proofs. We first derive the result for the case when, $n=mr$, that is, when all $m$ negative cliques each of order $r$ cover all the $n$ vertices of complete graph. 
\begin{theorem}\label{cpm2m}
The characteristic polynomial of $A(K^{m,r}_{mr})$ is given by 
$$\phi(K^{m,r}_{mr})=(1-\lambda)^{m(r-1)}(1-2r-\lambda)^{m-1}\Big(1+r(m-2)-\lambda\Big).$$
\end{theorem}
\begin{proof} With suitable relabelling of the vertices in $K^{m,r}_{mr},$ we have

$$A\Big(K^{m,r}_{mr}\Big)= \begin{bmatrix}
-A(K_r) & J & J & \cdots & J \\ J & -A(K_r) & J &\cdots & J \\ \vdots & \vdots & \vdots & \ddots & 
\vdots \\ J & J & J & \cdots & -A(K_r)
\end{bmatrix}_{{mr}\times {mr}}, $$

where, $A(K_r)$ denotes the adjacency matrix of a positive clique $K_r$, $J$ is all-one matrix of order $r$. Then,
$$A\Big(K^{m,r}_{mr}\Big)-\lambda I_{mr}= \begin{bmatrix}
Y & X & X & \cdots & X \\ X & Y & X &\cdots & X \\ \vdots & \vdots & \vdots & \ddots & 
\vdots \\ X & X & X & \cdots & Y
\end{bmatrix}_{mr} ,$$
where, $$Y=-A(K_r)-\lambda I_{r},\ \ X=J_{r},$$  and $J_{r}$ is all-one matrix of order $r.$

In the above matrix $A\Big(K^{m,r}_{mr}\Big)-\lambda I_{mr}$, subtract the last row from all the other rows. This produces
$$
\begin{bmatrix}Y-X & O & O & \dots & O & X-Y\\ O & Y-X & O & \dots & O & X-Y\\ O & O & Y-X & \dots & O & X-Y\\ \vdots & \vdots & \vdots & \ddots & \vdots & \vdots\\ O & O & O & \dots & Y-X & X-Y\\ X & X & X & \dots & X & Y\end{bmatrix},
$$
Now, add first $r-1$ columns to the last column. This produce the following lower triangular matrix, 
$$
\begin{bmatrix}Y-X & O & O & \dots & O & O\\ O & Y-X & O & \dots & O & O\\ O & O & Y-X & \dots & O & O\\ \vdots & \vdots & \vdots & \ddots & \vdots & \vdots\\ O & O & O & \dots & Y-X & O\\ X & X & X & \dots & X & \Big(Y+(m-1)X\Big)\end{bmatrix}.
$$

Hence, \begin{equation} \label{deteve}
\det\Bigg(A\Big(K^{m,r}_{mr}\Big)-\lambda I_{mr}\Bigg)=\det(Y-X)^{m-1}\det\Big(Y+(m-1)X\Big).
\end{equation}

Also, $$Y-X=-2A(K_r)-(\lambda+1)I_r.$$ 
The eigenvalues of $A(K_r)$ are given by $-1, (r-1)$ with the multiplicity $(r-1)$ , $1$, respectively \cite{bapat2010graphs}. Hence, the eigenvalues of the matrix, $Y-X$, are $(1-\lambda),  (-2r+1-\lambda)$ with multiplicities $(r-1)$ , $1$, respectively. As the determinant of a matrix is product of its eigenvalues including multiplicities, thus

$$\det(Y-X)=(1-\lambda)^{r-1}(1-2r-\lambda).$$
Next,
$$Y+(m-1)X=(m-2)A(K_r)+(m-1-\lambda)I_r.$$
The eigenvalues of $Y+(m-1)X$ are $(1-\lambda),\ \Big(1+r(m-2)-\lambda\Big)$ with multiplicity $(r-1), 1$, respectively.
Hence,
$$\det\Big(Y+(m-1)X\Big)=(1-\lambda)^{r-1}\Big(1+r(m-2)-\lambda\Big).$$
  
Thus, using Equation (\ref{deteve}) $$\phi(K^{m,r}_{mr})=\Big((1-\lambda )^{r-1}(1-2r-\lambda)\Big)^{m-1}(1-\lambda)^{r-1}\Big(1+r(m-2)-\lambda\Big)$$
$$=(1-\lambda)^{m(r-1)}(1-2r-\lambda)^{m-1}\Big(1+r(m-2)-\lambda\Big).$$
\end{proof}

\begin{corollary}
The determinant of $K^{m,r}_{mr}$ is given by $$(1-2r)^{(m-1)}\Big(1+r(m-2)\Big).$$
\end{corollary}

\begin{corollary} 
The eigenvalues of $K^{m,r}_{mr}$ are $1,(1-2r)$, and $\Big(1+r(m-2)\Big)$ with multiplicity $m(r-1), m-1$, and $1,$ respectively.
\end{corollary}

Next we give the inverse of the matrix $A(K^{m,r}_{mr})-\lambda I_{mr}$. It is used to get the characteristic polynomial for general case $A(K^{m,r}_{n})$.
\begin{lemma}\label{invm2m}
The inverse of $A(K^{m,r}_{mr})-\lambda I_{mr}$ is given by 
$$\frac{1}{\lambda+2r-1}\Bigg(\Big(\frac{1}{\lambda-1}\Big)I_m\otimes \Big(2 A(K_r)-(\lambda+2r-3)I_r\Big)-\frac{1}{\lambda+r(2-m)-1}J\Bigg),$$
where, $\lambda\ne 1,(1-2r)$, and $ \Big(1+r(m-2)\Big)$, $J$ is all-one matrix of order $mr$, and $\otimes$ denotes the tensor product of matrices. 
\end{lemma}
\begin{proof}
Using the same construction as in Theorem \ref{cpm2m}, we can write, $$A\Big(K^{m,r}_{mr}\Big)-\lambda I_{mr} = \Big(I_m\otimes (Y-X)\Big) + (1_{m\times m}\otimes X)=\Big(I_m\otimes (Y-X)\Big) + 1_{mr}1_{mr}^T.$$  Let $A_1=\Big(I_m\otimes (Y-X)\Big).$  Now, recall the Sherman-Morrison formula: If $A$ is a nonsingular square matrix and $ 1+v^{T}A^{-1}u\neq 0$ for some column vectors $u, v$, then $$(A+uv^{T})^{-1}=A^{-1}-{A^{-1}uv^{T}A^{-1} \over 1+v^{T}A^{-1}u}.$$

In order to find $A^{-1}_1$, we need to find $(Y-X)^{-1}$. By symmetry let $\alpha, \beta$ be the diagonal, non-diagonal entries of $(Y-X)^{-1},$ respectively. On solving the following two equations we get the values of $\alpha, \beta.$ 
\begin{eqnarray*}
-\alpha(\lambda+1)-2\beta(r-1)=1,\\
-\beta(\lambda+1)-2\alpha-2\beta(r-2)=0,
\end{eqnarray*}

we get, $$\alpha=\frac{-(\lambda+2r-3)}{(\lambda -1)(\lambda+2r-1)}, \beta=\frac{2}{(\lambda -1)(\lambda+2r-1)}.$$
Thus, $A^{-1}_1$ can be written as,
 $$A_1^{-1}= \frac{1}{(\lambda -1)(\lambda+2r-1)}\Bigg(I_m\otimes \Big(2 A(K_r)-(\lambda+2r-3)I_r\Big)\Bigg).$$

Also, $$A_1^{-1}1_{mr}1^{T}_{mr}A_1^{-1}=\frac{1}{(\lambda+2r-1)^2}\times J,$$

and $$1+1^{T}_{mr}A^{-1}_11_{mr}=\frac{\lambda+r(2-m)-1}{\lambda+2r-1},$$
where, $J$ is all-one matrix of order $mr$.

Hence, $$\Bigg(A\Big(K^{m,r}_{mr}\Big)-\lambda I_{mr}\Bigg)^{-1}= \frac{1}{(\lambda -1)(\lambda+2r-1)}\Bigg(I_m\otimes \Big(2 A(K_r)-(\lambda+2r-3)I_r\Big)\Bigg)-\frac{1}{(\lambda+2r-1)(\lambda+r(2-m)-1)}J$$

$$=\frac{1}{\lambda+2r-1}\Bigg(\Big(\frac{1}{\lambda-1}\Big)I_m\otimes \Big(2 A(K_r)-(\lambda+2r-3)I_r\Big)-\frac{1}{\lambda+r(2-m)-1}J\Bigg).
$$
\end{proof}

\begin{theorem}
The characteristic polynomial of $A(K^{m,r}_n)$ is given by 
\begin{eqnarray*}
(1-\lambda)^{m(r-1)}(1-2r-\lambda)^{m-1}\Bigg(\frac{-\lambda^2-r\Big(2+\lambda(2-m)-m\Big)+1}{\lambda+r(2-m)-1}\Bigg)^{n-mr-1}\\ \times\Bigg(n(1-2r-\lambda)+2r\Big(1+m(r-1)+\lambda\Big)-1+\lambda^2\Bigg).
\end{eqnarray*}
\end{theorem}
\begin{proof}
With suitable relabelling of the vertices in $K^{m,r}_n$, the matrix $A(K^{m,r}_n)-\lambda I_n$ can be written in the form $$A(K^{m,r}_n)=\begin{bmatrix} 
 A_1-\lambda I_{mr} & J \\
J^{T} & A_2-\lambda I_{n-mr} 
\end{bmatrix},$$ where, $A_1=A(K^{m,r}_{mr}),  A_2=A(K_{n-mr})$, $J$ is all-one matrix of order $(mr)\times(n-mr)$, and $J^T$ is the transpose of $J.$ By Schur complement formula (\cite{bapat2010graphs},p.4) we have,

$$ \det\Big(A(K^{m,r}_n)-\lambda I_n\Big)=\det(A_1-\lambda I_{mr})\times \det\Big((A_2-\lambda I_{n-mr}) -J^{T}(A_1-\lambda I_{mr})^{-1}J\Big).$$
Using Lemma \ref{invm2m}

$$J^T(A_1-\lambda I_{mr})^{-1}J=\frac{-mr}{\lambda+r(2-m)-1}J_1,$$

$$(A_2-\lambda I_{n-mr}) -J^{T}(A_1-\lambda I_{mr})^{-1}J=\Big(\frac{\lambda+2r-1}{\lambda+r(2-m)-1}\Big)K_{n-mr}+\Big(-\lambda+\frac{mr}{\lambda+r(2-m)-1}\Big)I_{n-mr}.$$

The eigenvalues of the above matrix are 
$$\frac{-\lambda^2-r\Big(2+\lambda(2-m)-m\Big)+1}{\lambda+r(2-m)-1},\ \  \frac{n(\lambda+2r-1)-2r(1+\lambda-m+mr)-\lambda^2+1}{\lambda+r(2-m)-1}$$ with the multiplicity $n-mr-1, \ 1,$ respectively.

By Theorem \ref{cpm2m} 
$$\det(A_1-\lambda I_{mr})=(1-\lambda)^{m(r-1)}(1-2r-\lambda)^{m-1}\Big(1+r(m-2)-\lambda\Big).$$

Hence, 
\begin{eqnarray*}
\phi\Big(A(K^{m,r}_n)\Big)=(1-\lambda)^{m(r-1)}(1-2r-\lambda)^{m-1}\Bigg(\frac{-\lambda^2-r\Big(2+\lambda(2-m)-m\Big)+1}{\lambda+r(2-m)-1}\Bigg)^{n-mr-1}\\ \times\Bigg(n(1-2r-\lambda)+2r\Big(1+m(r-1)+\lambda\Big)-1+\lambda^2\Bigg)
\end{eqnarray*}
\end{proof}

\begin{corollary}
The determinant of $A(K^{m,r}_n)$ is given by $$(1-2r)^{m-1}(-1)^{n-mr-1}\Bigg(n(1-2r)+2r\Big(1+m(r-1)\Big)-1\Bigg).$$
\end{corollary}

\begin{corollary}
The eigenvalues of $A(K^{m,r}_n)$ are $$1,\ (1-2r),\  \frac{(n-2r)\pm\sqrt{8mr-8r-4n-8mr^2+4+(n+2r)^2}}{2}$$ and the roots of the polynomial $$\Bigg(\frac{-\lambda^2-r\Big(2+\lambda(2-m)-m\Big)+1}{\lambda+r(2-m)-1}\Bigg),$$ with the multiplicity $m(r-1),\ (m-1),\ 1,\ n-mr-1$, respectively.
\end{corollary}

\section{Complete graph with negative cliques of different order} \label{cgnkdo}
In this section we consider the complete graph $G$ having disjoint negative cliques of different orders which cover the vertex-set of $G$. Assume that $G$ have $k$  negative cliques with order $n_1, n_2, \hdots, n_k,$ respectively. Let $n_1\le n_2\le \hdots \le n_k.$  Thus, the adjacency matrix of such a graph $G$ can be written as \begin{equation} A(G)=\begin{bmatrix}-A(K_{n_1}) & J_{12} & \hdots & J_{1k}\\ J_{12}^T & -A(K_{n_2}) & \hdots & J_{2k} \\ \vdots & \vdots & \ddots & \vdots \\ J_{1k}^T & J_{2k}^T & \hdots & -A(K_{n_k})\end{bmatrix},\end{equation}

where, $A(K_{n_i})$ denotes the adjacency matrix of the positive clique $K_{n_i}, i=1,\hdots,k$ and $J_{pq}$ denotes the all-one matrix of order $n_p\times n_q.$ To calculate the eigenvalues we use approach similar to in \cite{esser1980spectrum} for the complete multipartite graph. Note that, it is enough to investigate the eigenvalues of $A(G)-I_n$ in order to investigate the eigenvalues of $A(G).$ Indeed, $\lambda$ is an eigenvalue of $A(G)-I_n$ corresponding to an eigenvector $X\in {R}^n$ if and only if $\lambda+1$ is an eigenvalue of $G$ corresponding to the eigenvector $X.$ Observe that the diagonal blocks of $A(G)-I_n$ are $-J_{n_in_i}, i=1,\hdots,k,$ and the off diagonal blocks are same as that of $A(G).$

We first prove the following lemma which is used in the sequel.

\begin{lemma}\label{Lem:sec2}
Let \begin{equation}
N=\begin{bmatrix}
-n_1 & n_2  &\hdots  & n_k \\
n_1& -n_2  & \hdots  & n_k \\
 \vdots &\vdots & \ddots &\vdots \\
n_1 & n_2 & \hdots & -n_k
\end{bmatrix}
\end{equation} be a matrix of order $k\times k.$ Let $N_{\lambda}=N-\lambda I_k.$ Then $$\det(N_{\lambda})=\left[ \prod_{i=1}^{k}(-2n_i-\lambda)+\sum_{i=1}^{k}n_i\prod_{j=1,j \ne i}^{k}(-2n_j-\lambda)\right].$$
\end{lemma}
\begin{proof}
Let $\textbf{n}=[n_1 \,\, n_2 \,\, \hdots \,\, n_k]^T\in{R}^k.$ Then, $$\det(N_{\lambda})=\det\left(\begin{bmatrix} 1 & -\textbf{n}^T \\ 0_k & N_{\lambda}\end{bmatrix}\right)=\det\left(\begin{bmatrix}1 & -\textbf{n}^T\\ \textbf{1}_{k} & -2\, \mbox{diag}(\textbf{n})-\lambda I_k\end{bmatrix}\right).$$ Expanding the right hand side, the desired result follows.
\end{proof}

Now, we have the following theorem which completely characterizes the eigenvalues of $A(G)-I_n$, and hence the eigenvalues of $G$.

\begin{theorem}\label{Thm:Sec2}
Let $G$ be a complete graph on $n$ vertices with $k$ disjoint negative cliques of order $n_1, n_2,\hdots, n_k$ such that $n_1+n_2+\hdots+n_k=n.$ Suppose $\overline{n}_i, i=1,\hdots, t, \ \ t\leq k$ be the distinct numbers in the set $\{n_1,\hdots,n_k\}.$ Then,

 \begin{enumerate} \item[(a)] $0$ is an eigenvalue of $A(G)-I_n$ with algebraic multiplicity $n-k$ corresponding to eigenvectors $X=[X_1 \,\, X_2 \,\, \hdots \,\, X_k]^T, X_i\in{R}^{n_i}$ such that $\Sigma X_i=0$ for all $i,$ where $\Sigma X_i$ denotes the sum of entries in $X_i$. \item[(b)] $-2\overline{n}_i, i=1,\hdots, t$ are the  nonzero eigenvalues of $A(G)-I_n$ with multiplicity $m_i-1$ where $m_i$ is the number of distinct clusters in $G$ of order $\overline{n}_i.$ The other nonzero eigenvalues are the roots of the polynomial $1+p(\lambda)$ where $$p(\lambda)=\sum_{i=1}^t \frac{m_i\overline{n}_i}{-2\overline{n}_i-\lambda}.$$ Moreover, the eigenvectors corresponding to the nonzero eigenvalues of $A(G)-I_n$ are of the form $X=[\alpha_1\textbf{1}_{n_1}^T \,\, \alpha_2\textbf{1}_{n_2}^T \,\, \hdots \,\, \alpha_{k}\textbf{1}_{n_k}^T]^T$ where $0_k\neq \alpha=[\alpha_1 \,\, \alpha_2 \,\, \hdots \alpha_k]^T$ satisfies $N_{\lambda}\alpha=0.$ Such an $\alpha$ determines an eigenvector corresponds to the eigenvalue $\lambda$ for which $\lambda(\alpha_i-\alpha_j)=2(n_j\alpha_j -n_i\alpha_i), i,j=1,\hdots, k.$
\end{enumerate}
\end{theorem}

\begin{proof}\begin{enumerate}\item[(a)] Let $X=[X_1 \,\, X_2 \,\, \hdots \,\, X_k]^T, X_i\in{R}^{n_i}$ such that $\Big(A(G)-I_n\Big)X=0.$ Then for $i,j\in \{1, \hdots, k\},$ $$ \sum_{r\neq i, r=1}^k \Sigma X_r-\Sigma X_i = \sum_{r\neq j, r=1}^k \Sigma X_r -\Sigma X_j  =0.$$ This yields $\Sigma X_i=0$ for all $i=1,\hdots,k.$ Since dimension of the vector space $\{X_i\in{R}^{n_i} : \Sigma X_i=0\}$ over ${R}$ is $n_i-1,$ the desired result follows.

\item[(b)] Let $\lambda\neq 0$ and $\Big(A(G)-I_n\Big)X=\lambda X$ where $X=[X_1 \,\, X_2 \,\, \hdots \,\, X_k]^T, X_i\in{R}^{n_i}.$ For any $i,$ consider the vector $X_i$, any two entries of $X_i,$ say $x^{(i)}_p, x^{(i)}_q,$ satisfy \begin{equation}\label{eqn1:thm1}\lambda x^{(i)}_p = \sum_{r\neq i, r=1}^{k} \Sigma X_r - \Sigma X_i = \lambda x^{(i)}_q.\end{equation} Since, $\lambda\neq 0, X_i=\alpha_i \textbf{1}_{n_i}$ for some constant $\alpha_i$ for all $i=1,\hdots, k.$ Setting $X=[\alpha_1\textbf{1}_{n_1}^T \,\, \alpha_2\textbf{1}_{n_2}^T \,\, \hdots \,\, \alpha_{k}\textbf{1}_{n_k}^T]^T,$ by Equation (\ref{eqn1:thm1}) we have

\begin{equation}\label{eqn2:thm1} \lambda\alpha_i=\sum_{r\neq i, r=1}^k n_r\alpha_r-n_i\alpha_i.\end{equation} For any $j\neq i,$ similarly, we have
 \begin{equation}
 \lambda\alpha_j=\sum_{r\neq j, r=1}^k n_r\alpha_r-n_j\alpha_j.
 \end{equation} 

Adding these above two equations, we obtain $\lambda(\alpha_i-\alpha_j)=2(n_j\alpha_j -n_i\alpha_i)$ for any $i,j\in\{1,\hdots, k\}.$

In order to find all $\alpha_i, i=1,\hdots, k$ which satisfy Equation (\ref{eqn2:thm1}) for each $i,$ it gives the linear system $N_{\lambda}\alpha=0.$ Note that both $\lambda$ and $\alpha$ are unknown in this linear system and for the existence of a nonzero solution vector $\alpha,$ we must have $\det(N_{\lambda})=0.$ Thus, the nonzero eigenvalues of $A(G)-I_n$ are the roots of the polynomial $\det(N_{\lambda}).$ Now from Lemma \ref{Lem:sec2}, we have \begin{eqnarray*}\label{e10}
\det(T_\lambda) &=& \left[ \prod_{i=1}^{t}(-2\bar{n}_i-\lambda)^{m_i}+\sum_{i=1}^{t}\frac{m_i\bar{n}_i}{-2\bar{n}_i-\lambda}\prod_{j=1}^{t}(-2\bar{n}_j-\lambda)^{m_j}\right] \\
&=& \prod_{i}^{k}(-2\bar{n_i}-\lambda)^{m_i-1} \left[\prod_{i=1}^{t}(-2\bar{n}_i-\lambda)+\sum_{i=1}^{t}m_i\bar{n}_i \prod_{j=1,j\ne i}^{t}(-2\bar{n}_j-\lambda)\right].
\end{eqnarray*} Hence, the proof follows. 
  \end{enumerate}
\end{proof} 

\begin{lemma}\label{roots}
 Let $\lambda^{\star}_{1}> \lambda^{\star}_{2}> \hdots > \lambda^{\star}_{t-1}>\lambda^{\star}_{t}$ be the roots of polynomial $1+p(\lambda).$ Then \begin{equation}\label{eqn1:sec1}\lambda^{\star}_{1}>-2\bar{n}_{1}>\lambda^{\star}_{2} >-2\bar{n}_{2} \hdots >\lambda^{\star}_{t-1}>-2\bar{n}_{t-1}>\lambda^{\star}_{t}>-2\bar{n}_{t}.\end{equation} In general, if $\lambda_1\ge \lambda_2\ge \hdots\ge \lambda_{k-1}\ge \lambda_k$ are the nonzero eigenvalues of $A(G)-I_n$, then \begin{equation}\label{eqn2:sec1}\lambda_{1}\geq -2n_{1}\geq\lambda_{2}\geq -2n_{2} \hdots  \geq \lambda_{k-1}\geq -2n_{k-1}\geq \lambda_{k}\geq -2n_{k}.\end{equation}
\end{lemma}
\begin{proof}
Polynomial $p(\lambda)$ is continuous and strictly increasing  in interval $(-2\bar{n}_{i+1}, -2\bar{n}_{i})$. Also, $\lim\limits_{\lambda\to (-2\bar{n}_{i})^{-}}p(\lambda)=+\infty$ and $\lim\limits_{\lambda\to (-2\bar{n}_{i+1})^{+}}p(\lambda)=-\infty$ for $i=1,2\hdots t-1$. Hence, using intermediate value theorem there exists a root $\lambda^{*}_i$ of equation $1+p(\lambda)=0$ in interval $(-2\bar{n}_{i+1},-2\bar{n}_{i})$ for $i=1,2\hdots t-1$, satisfying $-2\bar{n}_{i}>\lambda^{*}_{i+1}>-2\bar{n}_{i+1}$. For $i=1$ $\lim\limits_{\lambda\to (-2\bar{n}_{1})^{+}}p(\lambda)=-\infty$ and $\lim\limits_{\lambda\to +\infty}p(\lambda)=0$. Again, using intermediate value theorem           $\lambda^{\star}_{1}>-2\bar{n}_1$ which proves (\ref{eqn1:sec1}). Similarly, (\ref{eqn2:sec1}) follows from Theorem \ref{Thm:Sec2}.
\end{proof}

\begin{corollary}
Let $\alpha_1\geq\alpha_2\geq\hdots\geq\alpha_n$ be the eigenvalues of $A$, and $\alpha^{*}_1>\alpha^{*}_2>\hdots>\alpha^{*}_{t-1}> \alpha^{*}_t$ be its non-zero non-integer eigenvalues. Then,
\begin{enumerate}
\item \begin{equation}\alpha^{\star}_{1}>-2\bar{n}_{1}+1>\alpha^{\star}_{2} >-2\bar{n}_{2}+1 \hdots >\alpha^{\star}_{t-1}>-2\bar{n}_{t-1}+1>\alpha^{\star}_{t}>-2\bar{n}_{t}+1.\end{equation}
\item \begin{equation}\alpha_{1}\geq -2n_{1}+1\geq\alpha_{2}\geq -2n_{2}+1 \hdots  \geq \alpha_{k-1}\geq -2n_{k-1}+1\geq \alpha_{k}\geq -2n_{k}+1.\end{equation}
\end{enumerate}
\end{corollary}
\begin{proof} It directly follows from the fact that $\alpha_i=\lambda_{i}+1 \ \forall \ i$ and Lemma \ref{roots}.
\end{proof}

\section{Regular Star Block Graph} \label{sbg}
In this section we calculate the eigenvalues of $r$-regular star block graph.

\begin{theorem} \label{lthm}
Let $G$ be a $r$-regular star block graph having $k$ blocks. If $l$ blocks are negative cliques for $l \le k $, then 
$$\phi(G)=l\Big(\phi(\tilde{K_r})\phi(\tilde{K}_{r-1})^{l-1}\phi({K}_{r-1})^{k-l}\Big)+(k-l)\Big(\phi({K_r})\phi(\tilde{K}_{r-1})^{l}\phi({K}_{r-1})^{k-l-1}\Big)+\lambda\Big( \phi(\tilde{K}_{r-1})^{l}\phi({K}_{r-1})^{k-l} \Big) ,$$ where, 
$\phi(\tilde{K_r})$ denotes the characteristic polynomial of a negative clique of order $r$. 
\end{theorem}
\begin{proof} Let $v$ be the only possible cut-vertex. Using the $\mathcal{B}$-partitions (\cite{singh2017characteristic}, Procedure 1) of $G$, when the cut-vertex $v$ associates with exactly one clique, it gives the following two product terms. 
$$l\Big(\phi(\tilde{K_r})\phi(\tilde{K}_{r-1})^{l-1}\phi({K}_{r-1})^{k-l}\Big)+(k-l)\Big(\phi({K_r})\phi(\tilde{K}_{r-1})^{l}\phi({K}_{r-1})^{k-l-1}\Big).$$ When the cut-vertex does not associates to any clique, it give the following product term.
$$\lambda\Big( \phi(\tilde{K}_{r-1})^{l}\phi({K}_{r-1})^{k-l} \Big).$$
Combining these product terms the result follows. 
\end{proof}
The eigenvalues of $K_n$ are $-1, (n-1)$ while the eigenvalues of $\tilde{K}_n$ are $1, 1-n$, with multiplicities $(n-1), 1$, respectively. Hence,
$$\phi(K_n)=(-1-\lambda)^{n-1}(n-1-\lambda),\ \phi(\tilde{K}_n)=(1-\lambda)^{n-1}(1-n-\lambda) $$

By Theorem \ref{lthm}
$$\phi(G)=l\Big(\phi(\tilde{K_r})\phi(\tilde{K}_{r-1})^{l-1}\phi({K}_{r-1})^{k-l}\Big)+(k-l)\Big(\phi({K_r})\phi({K}_{r-1})^{k-l-1}\phi(\tilde{K}_{r-1})^{l}\Big)+\lambda\Big( \phi(\tilde{K}_{r-1})^{l}\phi({K}_{r-1})^{k-l} \Big),$$
that is,
\begin{eqnarray*}
\phi(G)=\Big((1-\lambda)^{r-2}(2-r-\lambda)\Big)^{l-1}\Big((-1-\lambda)^{r-2}(r-2-\lambda)\Big)^{k-l-1} \times \\ \Bigg(l\Big((1-\lambda)^{r-1}(1-r-\lambda)(-1-\lambda)^{r-2}(r-2-\lambda)\Big)+(k-l)\Big((-1-\lambda)^{r-1}(r-1-\lambda)(1-\lambda)^{r-2}(2-r-\lambda)\Big)+\\ \lambda   \Big((1-\lambda)^{r-2}(2-r-\lambda)(-1-\lambda)^{r-2}(r-2-\lambda)\Big) \Bigg) 
\end{eqnarray*}

Thus, the eigenvalues of $G$ are $1, 2-r, -1, r-2$ with multiplicities $(r-2)(l-1),\ (l-1), \ (r-2)(k-l-1),\ (k-l-1)$, respectively, and the rest of the eigenvalues are given by the roots of the following polynomial.
\begin{eqnarray*}
\Bigg(l\Big((1-\lambda)^{r-1}(1-r-\lambda)(-1-\lambda)^{r-2}(r-2-\lambda)\Big)+(k-l)\Big((-1-\lambda)^{r-1}(r-1-\lambda)(1-\lambda)^{r-2}(2-r-\lambda)\Big)+\\ \lambda   \Big((1-\lambda)^{r-2}(2-r-\lambda)(-1-\lambda)^{r-2}(r-2-\lambda)\Big) \Bigg).
\end{eqnarray*}

\textbf{Acknowledgment.} The research of the first author was supported by the research fellowship of IIT
Jodhpur. The second author acknowledges support from the JC Bose Fellowship, Department of Science and
Technology, Government of India. The authors are grateful to Prof. Thomas Zaslavsky for his valuable comments and suggestions.

\bibliographystyle{plain}
    \bibliography{SB}
\end{document}